\documentclass[a4paper]{article}
\usepackage[margin=1.5in]{geometry}
\usepackage[utf8]{inputenc}
\usepackage{mathscinet}
\usepackage{comment}
\usepackage{graphicx}
\usepackage{booktabs}
\usepackage{hyperref}
\usepackage[numbers,sort&compress]{natbib}
\usepackage{enumerate}

\usepackage[dvipsnames]{xcolor}

\usepackage{tikz}
\usetikzlibrary{fit}
\usetikzlibrary{intersections}
\usetikzlibrary{decorations.pathreplacing}
\usetikzlibrary{calc}

\usepackage[vlined, ruled, linesnumbered]{algorithm2e}
\usepackage{algorithmicx}

\DontPrintSemicolon

\usepackage{amsmath}
\usepackage{amssymb}
\usepackage{amsthm}

\usepackage{thmtools}
\usepackage{thm-restate}

\usepackage{xspace}
\usepackage{authblk}
\usepackage[color=green!20]{todonotes}

\theoremstyle{plain}
\newtheorem{theorem}{Theorem}
\newtheorem{observation}[theorem]{Observation}
\newtheorem{lemma}[theorem]{Lemma}
\newtheorem*{thm*}{Theorem}
\newtheorem{proposition}[theorem]{Proposition}

\newtheorem*{eth-env}{Exponential Time Hypothesis}

\newcommand{\bip}{\textsc{Minimum Crossing Bipartite Matching}\xspace}
\def\msi{\textsc{Multicolored Subgraph Isomorphism}\xspace}

\def\tok{\textsc{Token Swapping}\xspace}
\def\ctok{\textsc{Colored Token Swapping}\xspace}
\def\stok{\textsc{Subset Token Swapping}\xspace}
\def\ts{\textsc{TS}\xspace}
\def\cts{\textsc{CTS}\xspace}
\def\sts{\textsc{STS}\xspace}
\def\sat{\textsc{3-Sat}\xspace}
\def\ecover{\textsc{Exact Cover by 3-Sets}\xspace}
\def\dm{\textsc{3-Dimensional Matching}\xspace}
\def\yes{\textsc{Yes}}
\def\no{\textsc{No}}
\newtheorem{claim}{Claim}
\newenvironment{inproof}[1]{\noindent {\bf Proof of Claim #1.}}{\hfill$\blacksquare$ \medskip}

\def\U{\textrm{U}\xspace}
\def\loc{\textrm{local}\xspace}
\def\glo{\textrm{global}\xspace}

\def\fpt{FPT\xspace}
\def\wone{$W[1]$\xspace}
\def\s{\boldsymbol{s}}
\DeclareMathOperator{\diam}{diam}
\DeclareMathOperator{\dist}{dist}
\DeclareMathOperator{\tw}{tw}
\DeclareMathOperator{\td}{td}
\newcommand{\xyswap}[2]{$(#1/#2)$-swap\xspace}
\newcommand{\xyswaps}[2]{$(#1/#2)$-swaps\xspace}

 \title{Complexity of Token Swapping and its Variants\footnote{Supported by the ERC grant PARAMTIGHT: ``Parameterized complexity and the search for tight complexity results'', no. 280152.}}
 \author[1]{\'Edouard Bonnet\thanks{\texttt{edouard.bonnet@dauphine.fr}}}
 \author[1]{Tillmann Miltzow\thanks{\texttt{t.miltzow@gmail.com}}}
 \author[1,2]{Pawe\l{} Rz\polhk{a}\.zewski \thanks{\texttt{p.rzazewski@mini.pw.edu.pl}}}
 \affil[1]{Institute for Computer Science and Control,
 Hungarian Academy of Sciences (MTA SZTAKI)}
 \affil[2]{Faculty of Mathematics and Information Science,
 		Warsaw University of Technology}		
\date{}

\begin{document}

\maketitle

\begin{abstract}
In the \textsc{Token Swapping} problem we are given a graph with a token placed on each vertex. 
Each token has exactly one destination vertex, and we try to move all the tokens to their destinations, using the minimum number of swaps, i.e., operations of exchanging the tokens on two adjacent vertices.
As the main result of this paper, we show that \textsc{Token Swapping} is \wone-hard parameterized by the length $k$ of a shortest sequence of swaps. 
In fact, we prove that, for any computable function $f$, it cannot be solved in time $f(k)n^{o(k / \log k)}$ where $n$ is the number of vertices of the input graph, unless the ETH fails.
This lower bound almost matches the trivial $n^{O(k)}$-time algorithm.

We also consider two generalizations of the \textsc{Token Swapping}, namely {\sc Colored Token Swapping} (where the tokens have colors and tokens of the same color are indistinguishable), and {\sc Subset Token Swapping} (where each token has a set of possible destinations). 
To complement the hardness result, we prove that even the most general variant, {\sc Subset Token Swapping}, is FPT in nowhere-dense graph classes.

Finally, we consider the complexities of all three problems in very restricted classes of graphs: graphs of bounded treewidth and diameter, stars, cliques, and paths, trying to identify the borderlines between polynomial and NP-hard cases.
\end{abstract}

\section{Introduction}\label{sec:intro}

In reconfiguration problems, we are interested to transform
a combinatorial or geometric object from one state to another, by performing a sequence of simple operations.
An important example is motion planning, where we want to move an object from one configuration to another. Elementary operations are usually translations and rotations. It turns out that motion planning can be reduced to the shortest path problem in some higher dimensional Euclidean space with obstacles~\cite{de2000computational}.

Finding the shortest flip sequence between any two triangulations of a convex polygon is a major open problem in computational geometry. 
Interestingly it is equivalent to a myriad of other reconfiguration problems of so-called Catalan structures~\cite{bose2009flips}.
Examples include: binary trees, perfect matchings of points in convex position, Dyck words, monotonic lattice paths, and many more.
Reconfiguring permutations under various constraints is heavily studied and usually called \emph{sorting}.

An important class of reconfiguration problems is a big family of problems in graph theory that involves 
moving tokens, pebbles, cops or robbers along the edges of a given graph, in order to reach some final configuration~\cite{WILSON197486,Calinescu2006,Fabila-Monroy2012,Graf2015, berlekamp1982winning, parsons1978pursuit, doi:10.1137/0208046, Demaine2015132, FoxEpsteinISAAC15}.
In this paper, we study one of them.

The \tok problem, introduced by Yamanaka {\em et al.} \cite{Yamanaka2014}, fits nicely into this long history of reconfiguration problems and can be regarded as a sorting 
problem with special constraints.

\begin{figure}[htbp]
 \centering
 \includegraphics{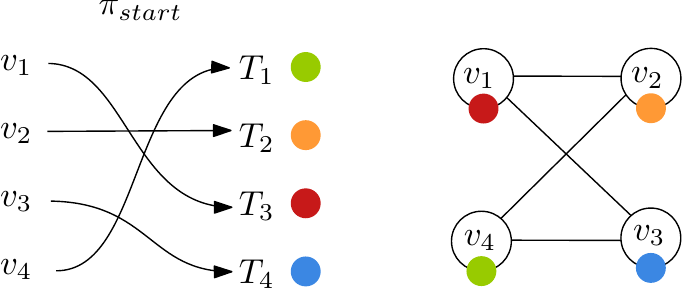}
 \caption{Every token placement can be uniquely described by a permutation.}
 \label{fig:Definition}
\end{figure}

The problem is defined as follows, see also Figure~\ref{fig:Definition}.
We are given an undirected connected graph with $n$ vertices $v_1,\dots,v_n$, a set of tokens $T = \{t_1,\ldots,t_n\}$ and 
two permutations $\pi_{\text{start}}$ and $\pi_{\text{target}}$.
These permutations are called start permutation and target permutation, respectively.
Initially vertex $v_i$ holds token $t_{\pi_{\text{start}}(i)}$.
In one step, we are allowed to \emph{swap} tokens on a pair of adjacent vertices, 
that is, if $v$ and $w$ are adjacent, $v$ holds the token $s$, and
$w$ holds the token $t$, then the swap between $v$ and $w$ results
in the configuration where $v$ holds $t$, $w$ holds $s$, and
all the other tokens stay in place.
The \tok problem asks if the target configuration can be reached in at most $k$ swaps. Thus, a solution for \tok is a sequence of edges, where the swaps take place. The solution is optimal if its length is shortest possible.
To see the correspondence to sorting note that every placement of tokens can be regarded as a permutation and the target permutation can be regarded as the sorted state.

Yamanaka {\em et al.} \cite{Yamanaka2014} observed that every instance of \tok has a solution, and its length is $O(n^2)$. Moreover, $\Omega(n^2)$ swaps are sometimes necessary. 
It is interesting to note that although the problem in its full generality was introduced only recently~\cite{Yamanaka2014}, some special cases were studied before in the context of sorting permutations with additional restrictions (see Knuth~\cite[Section 5.2.2]{Knuth} for paths, Pak \cite{PAK1999} for stars, Cayley \cite{Cayley} for cliques, and Heath and Vergara~\cite{HeathV03} for squares of a path).
Recently the problem was also solved for a special case of complete split graphs (see Gaku {\em et. al.}~\cite{Yasui-SIGAL}). It is also worth mentioning that  a very closely related concept of sorting permutations using cost-constrained transitions was considered by Farnoud, Chen, and Milenkovic \cite{FarnoudM12}, and Farnoud and Milenkovic\cite{FarnoudM10}.

The computational complexity of \tok was investigated by Miltzow {\em et al.}~\cite{MiltzowNORTU16}. They show that the problem is NP-complete and APX-complete. 
Moreover, they show that any algorithm solving \tok in time $2^{o(n)}$ would refute the Exponential Time Hypothesis (ETH) \cite{ImpagliazzoPaturi}.
The results of Miltzow {\em et al.}~\cite{MiltzowNORTU16} carry over also to a generalization of \tok, called \ctok, first introduced by Yamanaka {\em et al.}~\cite{DBLP:conf/wads/YamanakaHKOSUU15}.
In this problem, vertices and tokens are partitioned into color classes.
For each color $c$, the number of tokens colored $c$ equals the number of vertices colored $c$. The question is whether $k$ swaps are enough to reach a configuration in which each vertex contains a token of its own color.
\tok corresponds to the special case where each color class comprises exactly one token and one vertex. NP-hardness of \ctok was first shown by Yamanaka {\em et al.}~\cite{DBLP:conf/wads/YamanakaHKOSUU15}, even in the case that only $3$ colors exist.

We introduce  \stok, which is an even further generalization of \tok. Here a function $D: T \rightarrow 2^V$ specifies the set $D(t)$ of possible  destinations for the token $t$. We ask if $k$ swaps are enough to reach a configuration, when each token $t$ is placed on a vertex from $D(t)$.
Observe that \stok also generalizes \ctok.
It might happen that there is no satisfying swapping sequence at all to this new problem.
Though, this can be checked in polynomial time by deciding if there is a perfect matching in the bipartite token-destination graph. Thus we shall always assume that we have a satisfiable instance.

In this paper we continue and extend the work of Miltzow {\em et al.} \cite{MiltzowNORTU16}. They presented a very simple algorithm which solves the instance of \tok in $n^{O(k)}$ time and space, where $k$ denotes the number of allowed swaps. In Section \ref{sec:algorithms} we show that this algorithm can be easily generalized to \ctok and \stok.
Next, we present a slightly slower exact algorithm, whose advantage is only polynomial (in fact, only slightly super-linear) space complexity.

The algorithm by Miltzow {\em et al.} \cite{MiltzowNORTU16} shows that \tok is in XP. A natural next step is to investigate  whether the problem can be solved in FPT time (i.e., $f(k) \cdot n^{O(1)}$, for some function $f$). There is some evidence indicating that this could be possible. First, observe that if more than $2k$ tokens are misplaced, then one can immediately answer that we deal with a \no-instance, as each swap involves exactly two tokens.
Further, one can safely remove all vertices from the graph that are at distance more than $k$ from all misplaced tokens. This preprocessing yields an equivalent instance, where every connected component has \emph{diameter} $O(k^2)$. Thus for bounded maximum degree $\Delta$ each component has size $f(k)$, for some function $f$.
The connected components of $f(k)$ \emph{size} can be solved separately by exhaustively guessing (still in FPT time) the number of swaps to perform in each of them. Moreover, even the generalized \stok problem is FPT in $k + \Delta$ (see Proposition~\ref{prop:fpt-delta}). For those reasons, one could have hoped for an FPT algorithm for general graphs.
However, as the main result of this paper, we show 
in Section~\ref{sec:param} that this is not possible.
\begin{restatable}[Parameterized Hardness]{theorem}{paraHard}\label{th:param-k-ts}
\tok is {\wone}-hard, parameterized by the number $k$ of allowed swaps.
Moreover, assuming the ETH, for any computable function $f$, \tok cannot be solved in time $f(k)(n+m)^{o(k / \log k)}$ where $n$ and $m$ are respectively the number of vertices and edges of the input graph.
\end{restatable}
Observe that this lower bound shows that the simple $n^{O(k)}$-time algorithm is almost best possible.
It is worth mentioning that the parameter for which we show hardness is in fact \emph{number of swaps} + \emph{number of initially misplaced tokens} + \emph{diameter of the graph}, which matches the reasoning presented in the previous paragraph.

To show the lower bound, we introduce handy gadgets called \emph{linkers}. They are simple and can be used to give a significantly simpler proof of the lower bounds given by  Miltzow {\em et al.}~\cite{MiltzowNORTU16}. 
One might also use them to establish a simpler and potentially stronger inapproximability result.

Since there is no FPT algorithm for \tok (parameterized by the number $k$ of swaps), unless \fpt = \wone, a natural approach is to try to restrict the input graph classes, in hope to obtain some positive results. Indeed, in Section \ref{sec:nowheredense} we show that FPT algorithms exist, if we restrict our input to the so-called {\em nowhere-dense graph classes}.
\begin{restatable}[FPT in nowhere dense graphs]{theorem}{FPTnowheredense}\label{th:param-k-lbtw-ts}
\stok is FPT parameterized by $k$ on nowhere-dense graph classes.
\end{restatable}
The notion of nowhere-dense graph classes has  been introduced as a common generalization of several previously known notions of sparsity in graphs such as planar graphs, graphs with forbidden (topological) minors, graphs with (locally) bounded treewidth or graphs with bounded maximum degree.

Grohe, Kreutzer, and Siebertz \cite{Grohe2014} proved that every property definable as a first-order formula $\varphi$ is solvable in $O(f(|\varphi|,\varepsilon)\, n^{1+\varepsilon})$ time on nowhere-dense classes of graphs, for every $\varepsilon>0$. We use this meta-theorem to show the existence of an FPT time algorithm for \stok, restricted to nowhere-dense graph classes. In particular, this implies the following results.

\begin{restatable}{cor}{FPTcor}\label{cor:param-k}
\stok is FPT
\begin{enumerate}[(a)]
\item\label{k+tree} parameterized by $k + \tw(G)$,
\item parameterized by $k$ in planar graphs.
\end{enumerate}
\end{restatable}
\bigskip

It is often observed that NP-hard graph problems become tractable on classes of graphs with bounded treewidth (or, at least, with bounded tree-depth; see Ne\v{s}et\v{r}il and Ossona de Mendez \citep[Chapter 10]{Sparsity} for the definition and some background of tree-depth and related parameters). It is not uncommon to see FPT algorithms running in time $f(\tw)n^{O(1)}$ (or $f(\td)n^{O(1)}$) or XP algorithms running in time $n^{f(\tw)}$ (or $n^{f(\td)}$), for some computable functions $f$.
Especially, in light of Corollary~\ref{cor:param-k}~(\ref{k+tree}), we want to know if there exists an algorithm that runs in polynomial time for constant treewidth.
In Section \ref{sec:almosttrees} we rule out the existence of such algorithms by showing that \tok remains NP-hard when restricted to graphs with tree-depth $4$ (treewidth and pathwidth $2$; diameter $6$; distance $1$ to a forest).

\begin{restatable}[Hard on Almost Trees]{theorem}{HardAlmostTrees}\label{th:param-twd-ts}
 \tok remains \emph{NP}-hard even when both the treewidth and the diameter of the input graph are constant, and cannot be solved in time $2^{o(n)}$, unless the ETH fails.
 \end{restatable}

 \bigskip
 
Table \ref{conc:summary1} shows the current state of our knowledge about the parameterized complexity of \tok (TS), \ctok (CTS), and \stok (STS) problems, for different choices of parameters.

\begin{table}[h!]
\centering
\begin{tabular}{lccccc}
\toprule
     & $k + \Delta$  &  $k+\diam$     & $k$, nowhere-dense  
     & $\tw + \diam$     \\
     &   &       &  /   $k + \tw$      &      \\
\midrule
\ts  & FPT (\cite{MiltzowNORTU16}) & \textbf{W[1]-h} (Th~\ref{th:param-k-ts}) &  FPT      & \textbf{paraNP-c} (Th~\ref{th:param-twd-ts}) \\
\cts & FPT &W[1]-h & FPT                  & paraNP-c \\
\sts & \textbf{FPT~} (Prop~\ref{prop:fpt-delta}) &W[1]-h & \textbf{FPT} (Th~\ref{th:param-k-lbtw-ts})   & paraNP-c  \\
\bottomrule
\end{tabular}
\vspace{0.2cm}
\caption{The parameterized complexity of \tok, \ctok, and \stok.}
\label{conc:summary1}
\end{table}

While we think that our results give a fairly detailed view on the complexity landscape of \tok, we also want to point out that our reductions are significantly simpler than those by Miltzow {\em et al.}~\cite{MiltzowNORTU16}.

Since the investigated problems seem to be immensely intractable, in Section \ref{sec:trees} we investigate their complexities in very restricted classes of graphs, namely cliques, stars, and paths. We focus on finding the borderlines between easy (polynomially solvable) and hard (NP-hard) cases. The summary of these results is given in Table \ref{conc:summary2}. Observe that on cliques \tok is in P, while \ctok (and thus also \stok) is NP-hard. On the other hand, on stars \ctok (and thus also \tok) is in P and \stok is NP-hard.

\begin{table}[h!]
\centering
\begin{tabular}{lccccc}
\toprule
        & trees  & cliques  & stars   & paths  \\
\midrule
TS  & ?  & P (see \cite{MiltzowNORTU16})     & P (see \cite{MiltzowNORTU16})  & P (see \cite{MiltzowNORTU16})       \\
CTS & ? & \textbf{NP-c} (Th~\ref{th:cliques-cts}) &\textbf{P} (Th~\ref{th:stars-cts})  & \textbf{P} (Th~\ref{th:paths-cts})      \\
STS& NP-c & {NP-c} & \textbf{NP-c} (Th~\ref{th:stars-sts}) & NP-c~\cite{Guspiel}  \\
\bottomrule
\end{tabular}
\vspace{0.2cm}
\caption{The complexity of \tok (TS), \ctok (CTS), and \stok (STS) on very restricted classes of graphs.}
\label{conc:summary2}
\end{table}

The paper is concluded with several open problems in Section~\ref{sec:conclusion}.

\section{Preliminaries}\label{sec:preliminaries}
Yamanaka {\em et al.} \cite{Yamanaka2014} showed that in every instance of \tok, the length of the optimal solution is $O(n^2)$ and this bound is asymptotically tight for paths. Here we show that long induced paths are the only structures forcing solutions of superlinear length.

\begin{proposition} \label{prop:pr-free}
The length of the optimal solution for \tok in an $n$-vertex $P_{r+1}$-free graph $G$ is at most $r \cdot n$.
\end{proposition}

  \begin{proof}
  We can assume that $G$ is connected, since otherwise we can solve the problem on connected components separately.
  Let $P$ be the longest path in $G$ and let $v$ be its end-vertex. Observe that $G - v$ is connected (otherwise $P$ is not longest) and $P_{r+1}$-free. First, we move the token with destination $v$ to this vertex, which requires at most $\diam(G) \leqslant r$ swaps. Then we can recursively continue with the graph $G-v$ (we never touch $v$ again). Such a solution has length at most $r \cdot n$.
  \end{proof}

Note that this bound is asymptotically tight -- to see this, consider a graph, whose every connected component is isomorphic to $P_r$ and has the reverse permutation of tokens (if we want to have our instance connected, we can add one additional vertex, adjacent to one of the end-vertices of each path, and put a well-placed token on it).
Moreover, we observe that the bound from Proposition \ref{prop:pr-free} holds also for \ctok and \stok problems. Indeed, we can fix one destination for each of the tokens (by choosing a perfect matching in the token-destination graph) to obtain an instance of \tok, whose solution is also the solution for the original problem.

For a token $t$, let $\dist(t)$ denote the distance from the position of $t$ to its destination. 
For an instance $I$ of \tok, we define $L(I) := \sum_{t} \dist(t)$, i.e., the sum of distances to the destination over all the tokens. 
Clearly, after performing a single swap, $\dist(t)$ may change by at most $1$. 
We shall also use the following classification of swaps: for $x,y \in \{-1,0,1\}$, $x \leq y$, by a {\em \xyswap{x}{y}} we mean a swap, in which one token changes its distance by $x$, and the other one by $y$. 
Intuitively, \xyswaps{-1}{-1} are the most ``efficient'' ones, thus we will call them {\em happy swaps}. 
Since each swap involves two tokens, we get the following lower bound.

\begin{proposition}[\cite{MiltzowNORTU16}] \label{prop:lower}
The length of the optimal solution for an instance $I$ of \tok is at least $L(I)/2$.
Besides, it is \emph{exactly} $L(I)/2$ if and only if there is a solution using happy swaps only.
\end{proposition}

When designing algorithms, especially for computationally hard problems, it is natural to ask about lower bounds. However, the standard complexity assumption used for distinguishing easy and hard problems, i.e., P $\neq$ NP, is too weak to tell us something meaningful about possible complexities of algorithms. The stronger assumption that is typically used for this purpose is the so-called {\em Exponential Time Hypothesis} (usually referred to as the ETH), formulated by Impagliazzo and Paturi \cite{ImpagliazzoPaturi}. We refer the reader to the survey by Lokshtanov, Marx, and Saurabh for more information about ETH and conditional lower bounds \cite{LokshtanovMS11}.
The version we present below (and is most commonly used) is not the original statement of this hypothesis, but its weaker version (see also Impagliazzo, Paturi, and Zane \cite{DBLP:journals/jcss/ImpagliazzoPZ01}).

\begin{eth-env}[Impagliazzo and Paturi \cite{ImpagliazzoPaturi}]
There is no algorithm solving every instance of \sat with $N$ variables and $M$ clauses in time $2^{o(N+M)}$.
\end{eth-env}

\section{Algorithms}\label{sec:algorithms}
First, we prove that \stok (and therefore also \ctok as its restriction) is FPT in $k + \Delta$, where $k$ is the number of allowed swaps, and $\Delta$ is the maximum degree of the input graph. This generalizes the observation of Miltzow {\em et al.} \cite{MiltzowNORTU16} for \tok.
Furthermore,  we show that the simple algorithm for \tok, presented by Miltzow {\em et al.} \cite{MiltzowNORTU16}, carries over to the generalized problems, i.e., \ctok and \stok. 
At last, we will present an algorithm that has polynomial space complexity.

\begin{proposition} \label{prop:fpt-delta}
\stok is FPT in $k + \Delta$ and admits a kernel of size $2k + 2k^2 \cdot \Delta^k$.
\end{proposition}
\begin{proof}
Let $I$ be an instance of \stok on a graph $G$ with maximum degree $\Delta$ and suppose $I$ has a solution $\s$ of length at most $k$.

Let $V_m$ be the set of such vertices $v$ of $G$, that the token initially placed on $v$ does not accept $v$ as its destination.
First, observe that every vertex from $V_m$ has to be involved in some swap in $\s$. Thus  we can assume that $|V_m|\leq 2k$ (otherwise we immediately report a \no-instance).

Let $E'$ be the set of edges that appear in $\s$ and let $G'$ be the subgraph of $G$ induced by $E'$. Consider a connected component $C$ of $G'$.
Suppose first that the vertex set of $C$ does not contain any vertex from $V_m$. Observe that the sequence $\s'$ obtained from $\s$ by removing all edges from $C$ is also a solution for $I$ of length at most $k$.
So, without loss of generality, every connected component $C$ of $G'$ contains a vertex from $V_m$, and has at most $k$ edges. Let $G''$ be the subgraph of $G$  induced by the vertices at distance at most $k$ from $V_m$ (we find it by running a breadth-first search, starting from $V_m$). We observe that every vertex incident to an edge in $E'$ is in $G''$. Thus the instance $I'$ of \stok, restricted to $G''$, is equivalent to $I$. Note that the maximum degree of $G''$ is at most $\Delta$, and the number of vertices in $G''$ is at most $2k + 2k^2\Delta^{k}$. Thus $I'$ is a kernel for $I$.
\end{proof}

Miltzow {\em et al.} \cite{MiltzowNORTU16} show that an optimal solution for \tok can be found by performing a breath-first-search on the {\em configuration graph}, i.e. a graph, whose vertices are all possible configurations of tokens on vertices, and two configurations are adjacent when we can obtain one from another with a single swap. We observe that the same approach works for  \ctok and  \stok, the only difference is that we terminate on any feasible target configuration.

\begin{proposition} \label{prop:exact}
Let $G$ be a graph with $n$ vertices, and let $k$ be the maximum number of allowed swaps. The \ctok and the \stok problems on $G$ can be solved in time:
\begin{itemize}
\item $O(n! \cdot n^2) = 2^{O(n \log n)}$,
\item $n^{O(k)}=2^{O(k \log n)}$,
\end{itemize}
using exponential space. \qed
\end{proposition}

The main drawback of such an approach is an exponential space complexity. Here we show the following complementary result, inspired by the ideas of Savitch~\cite{Savitch70}.

\begin{theorem}
Let $G$ be a graph with $n$ vertices, and let $k$ be the maximum number of allowed swaps. \stok on $G$ can be solved in time $2^{O(n \log n  \log k)}=2^{O(n \log^2 n)}$ and space $O(n \log n \log k)=O(n \log^2 n)$.
\end{theorem}
\begin{proof}
Consider the algorithm \texttt{Reach} (see Algorithm \ref{alg:reach}). It is easy to verify that it returns {\em true} if the configuration $\pi_s$ can be reached from the configuration $\pi_0$ with \textbf{exactly} $k$ swaps, and {\em false} otherwise.

\begin{algorithm}[htb]
\small
\caption {Reach($G,\pi_0,\pi_s,k$)}
\label{alg:reach}
\KwIn{$G=(V,E)$ -- a graph, $\pi_0,\pi_s$ -- configurations of tokens on $G$, integer $k \geq 0$}
\If{$k = 0$}
{
	\lIf{$\pi_0 = \pi_s$}{\Return {\em true}}
	\lElse{\Return {\em false}}	
}
\If{$k = 1$}
{
	\ForEach{$e \in E$}
	{
		\If{$\pi_s$ can be obtained from $\pi_0$ with a swap on $e$}{\Return {\em true}}
	}
			\Return {\em false}
}
\Else
{
	\ForEach{configuration $\pi'$ of tokens on $G$}
	{		
		\If {$\operatorname{Reach}(G,\pi_0,\pi',\lceil k/2 \rceil)=true$ \textbf{and} $\operatorname{Reach}(G,\pi',\pi_s,\lfloor k/2 \rfloor)=true$}
		{
		\Return {\em true}
		}		
	}
	\Return {\em false}
}
\end{algorithm}
The depth of the recursion is $O(\log k)$.  The configurations can be generated with polynomial delay, using only linear (in $n$) memory. Thus the time complexity of the algorithm is $n!^{\log k} \cdot n^{O(1)} = 2^{O(n \log n \log k)}$. The space needed to keep track of the recursive stack is $O(n \log n \log k)$. Recall that $k = O(n^2)$ -- otherwise we immediately report a \yes-instance.

To use the algorithm for  \stok, we can enumerate all possible target configurations in $n! \cdot n^{O(1)} = 2^{O(n \log n)}$ time and polynomial space, and then solve the instance of \tok for each of them.
\end{proof}

\section{Lower bounds on parameterized Token Swapping}\label{sec:param}
Let us start by defining an auxiliary problem, called \msi (also known as \textsc{Partitioned Subgraph Isomorphism}; see Figure~\ref{fig:SubgraphIso}).

\begin{figure}
 \centering
 \includegraphics{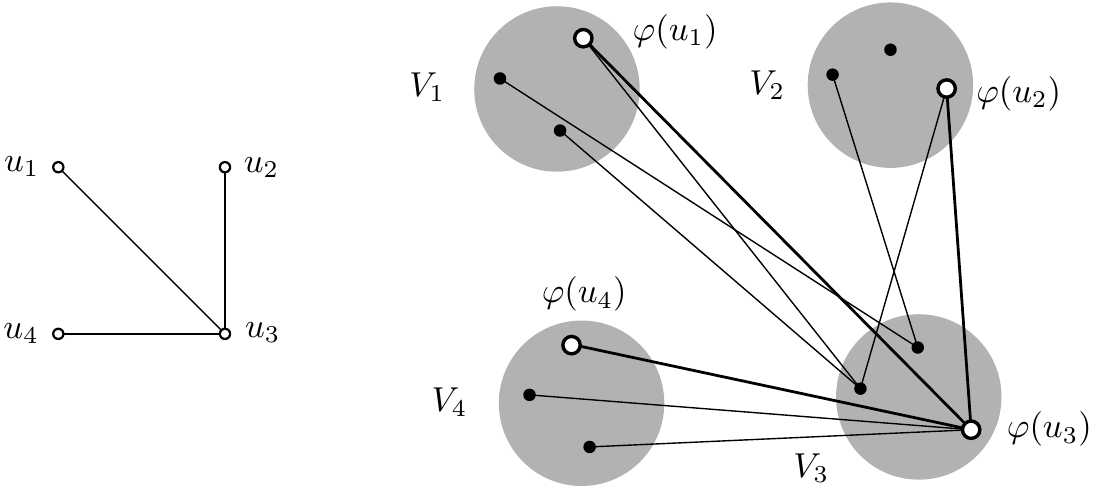}
 \caption{On the left is the pattern graph $P$; on the right, the host graph $H$. 	We indicate the image of $\varphi$ with white vertices. To keep the example small, we did not make $P$ $3$-regular.
}
 \label{fig:SubgraphIso}
\end{figure}

In \msi, one is given a host graph $H$ whose vertex set is partitioned into $k$ color classes $V_1 \uplus V_2 \uplus \ldots \uplus V_k$ and a pattern graph $P$ with $k$ vertices: $V(P) = \{u_1,\ldots,u_k\}$.
The goal is to find an injection $\varphi: V(P) \rightarrow V(H)$ such that
$u_iu_j\in E(P)$ implies that $\varphi(u_i)\varphi(u_j) \in E(H)$ and $\varphi(u_i) \in V_i$ for all $i,j$.
Thus we can assume that each $V_i$ forms an independent set.
Further, we assume without loss of generality that $E(V_i,V_j) := \{ ab \in E(H) \colon a \in V_i, b \in V_j\}$ is non-empty if and only if $u_iu_j \in E(P)$. 
In other words, we try to find $k$ vertices $v_1 \in V_1$, $v_2 \in V_2$, $\ldots$, $v_k \in V_k$ such that, for any $i < j \in [k]$,\footnote{For an integer $p$, by $[p]$ we denote the set $\{1,\ldots,p\}$.} there is an edge between $v_i$ and $v_j$ if and only if $E(V_i,V_j)$ is non-empty. 
The \wone -hardness of \msi problem follows from the \wone -hardness of the \textsc{Multicolored Clique}. Marx \cite{Marx10} showed that assuming the ETH, \msi cannot be solved in time $f(k)(|V(H)|+|E(H)|)^{o(k / \log k)}$, for any computable function $f$, even when the pattern graph $P$ is $3$-regular and bipartite (see also Marx and Pilipczuk \cite{MarxP15}).
In particular, $k$ has to be an even integer since $|E(P)|$ is exactly $3k/2$. 
We finally assume that for every $i \in [k]$ it holds that $|V_i|=t$, by padding potentially smaller classes with isolated vertices. 
This can only increase the size of the host graph by a factor of $k$, and does not create any new solution nor destroy any existing one.

Now we are ready to prove the following theorem.

\paraHard*

\begin{proof}
To show parameterized hardness of \tok, we introduce a very handy \emph{linker gadget}.
This gadget has a robust and general ability to link decisions.
As such, it permits to reduce from a wide range of problems.
Its description is short and its soundness is intuitive. 
Because it yields very light constructions, we can rule out fairly easily unwanted swap sequences.
We describe the {linker gadget} and provide some intuitive reason why it works (see Figure~\ref{fig:linker}).
\paragraph{\textbf{Linker gadget.}}

Given two integers $a$ and $b$, the linker gadget $L_{a,b}$ contains a set of $a$ vertices, called \emph{finishing set} and a path on $a$ vertices, that we call \emph{starting path}.
% Further, there are $b$ private paths all of length $a$ as well. 
% The number $a$ represents the number of other linker gadgets that are linked to this gadget and $b$ represents the number of possible choices. 
The tokens initially on vertices of the finishing set are called \emph{local tokens}; they shall go to the vertices of the starting path in the way depicted in Figure~\ref{fig:linker}.
The tokens initially on vertices of the starting path are called \emph{global tokens}.
Global tokens have their destination in some other linker gadget.
To be more specific, their destination is in the finishing set of another linker.

We describe and always imagine the finishing set and the starting paths \emph{to be ordered from left to right}.  
Below the finishing set and to the left of the starting path, stand $b$ disjoint induced paths, each with $a$ vertices, arranged in a grid, see Figure~\ref{fig:linker}.  
We call those paths \emph{private paths}.
The \emph{private tokens} on private paths are already well-placed.
Every vertex in the finishing set is adjacent to all private vertices below it and the leftmost vertex of the starting path is adjacent to all rightmost vertices of the private paths.

\begin{figure}[h!]
\centering
\includegraphics{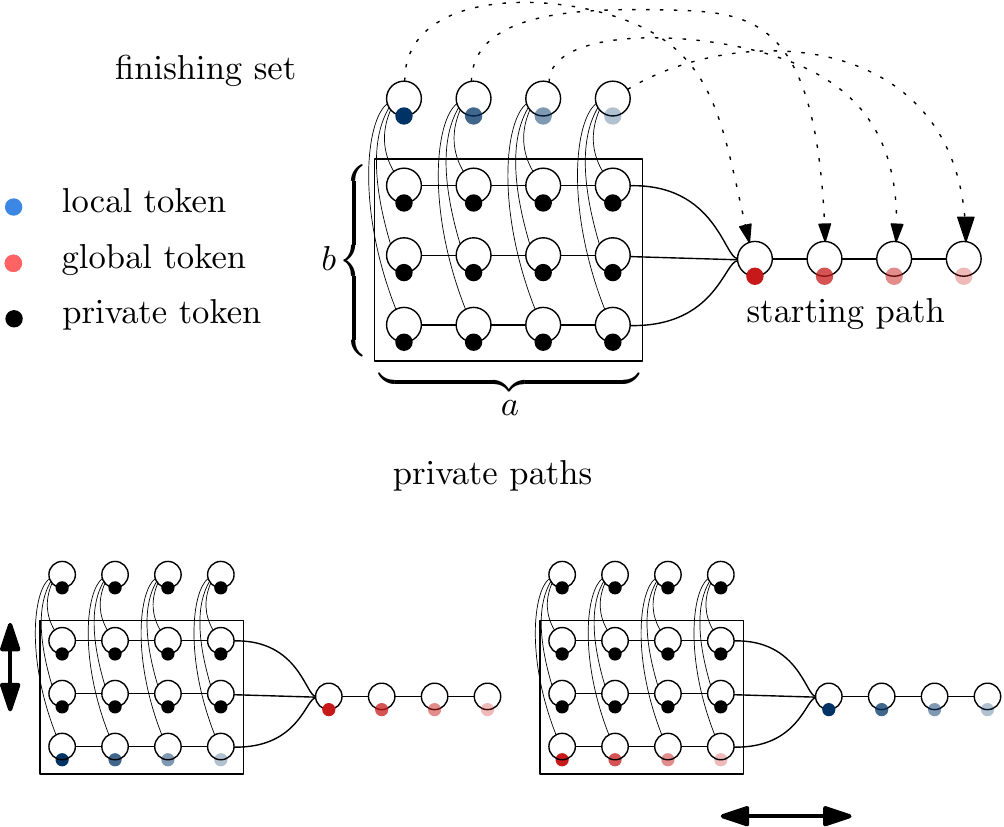}

\caption{The linker gadget $L_{a,b}$. Black (private) tokens are initially properly placed.
Dashed arcs represent destinations of tokens of the finishing set (they all go to the starting path).
In the intended solution, all local tokens are moved to a single private path (bottom left). 
Next, they are swapped with the tokens on the starting path (bottom right). 
The global tokens go to that private path.}
\label{fig:linker}
\end{figure}

For local tokens to go to the starting path, they must go through a private path.
As its name suggests, the linker gadget aims at linking the choice of the private path used for every local token. 
Intuitively, the only way of benefiting from $a^2$ happy swaps between the $a$ local tokens and the $a$ global tokens is to use a unique private path (note that the destination of the global tokens will make those swaps happy).
That results in a kind of configuration as depicted in the bottom right of Figure~\ref{fig:linker}, where each global token is in the same private path.
The fate of the global tokens has been linked. 

\paragraph{\textbf{Construction.}}
We present a reduction from \msi with cubic pattern graphs to \tok where the number of allowed swaps is linear in $k$. Let $(H,P)$ be an instance of \msi.
For any color class $V_i=\{v_{i,1},v_{i,2},\ldots,v_{i,t}\}$ of $H$, we add a copy of the linker $L_{3,t}$ that we denote by $L_i$.
We denote by $j_1<j_2<j_3$ the indices of the neighbors of $u_i$ in the pattern graph $P$.
The linker $L_{i}$ will be linked to $3$ other gadgets and it has $t$ private paths (or \emph{choices}).
The finishing set of $L_i$ contains, from left to right, the vertices $a(i,j_1)$, $a(i,j_2)$, and $a(i,j_3)$.
We denote the tokens initially on the vertices $a(i,j_1)$, $a(i,j_2)$, and $a(i,j_3)$ by $\loc(i,j_1),\loc(i,j_2),\loc(i,j_3)$, respectively.

\begin{figure}[h!]
\centering
\includegraphics{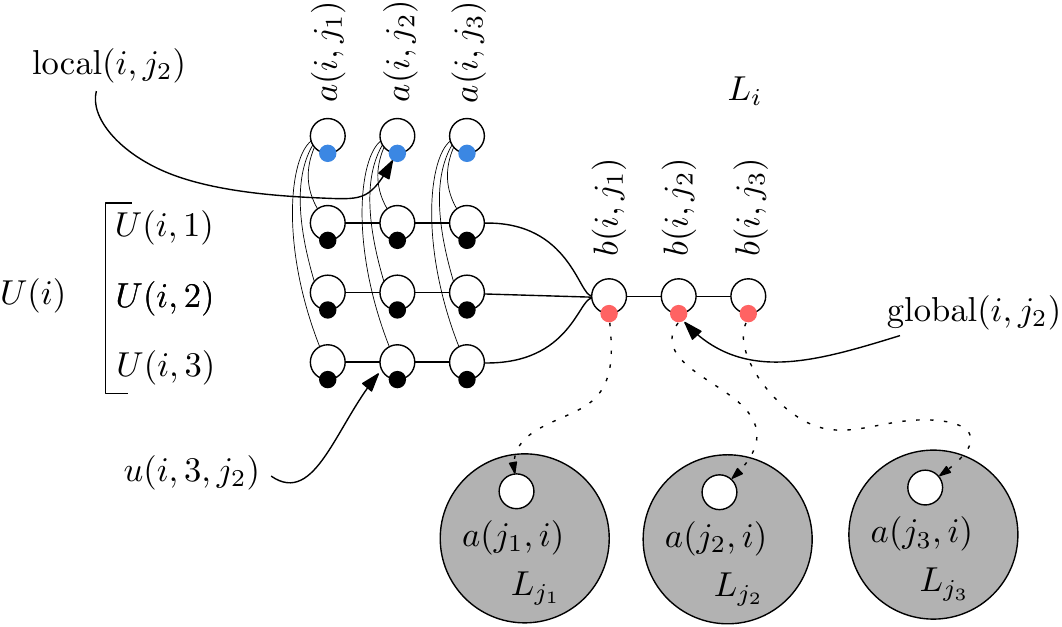}
\caption{The different labels for tokens, vertices, and sets of vertices.}
\label{fig:linker2}
\end{figure}

The starting path contains, from left to right, vertices $b(i,j_1)$, $b(i,j_2)$, and $b(i,j_3)$ with tokens $\glo(i,j_1)$, $\glo(i,j_2)$, and $\glo(i,j_3)$.

For each $p \in [3]$, $\loc(i,j_p)$ shall go to vertex $b(i,j_p)$, whereas $\glo(i,j_p)$ shall go to $a(j_p,i)$ in the gadget $L_{j_p}$.
Observe that the former transfer is internal and may remain within the gadget $L_i$, while the latter requires some interplay between the gadgets $L_i$ and $L_{j_p}$.
For any $h \in [t]$, by $\U(i,h)$ we denote the $h$-th private path.
This path represents the vertex $v_{i,h}$.
The path $\U(i,h)$ consists of, from left to right, vertices $u(i,h,j_1)$, $u(i,h,j_2)$, $u(i,h,j_3)$.
We set $\U(i) := \bigcup_{h \in [t]} \U(i,h)$.
Initially, all the tokens placed on vertices of $\U(i)$ are already well placed.

\begin{figure}[htbp]
\centering
\includegraphics{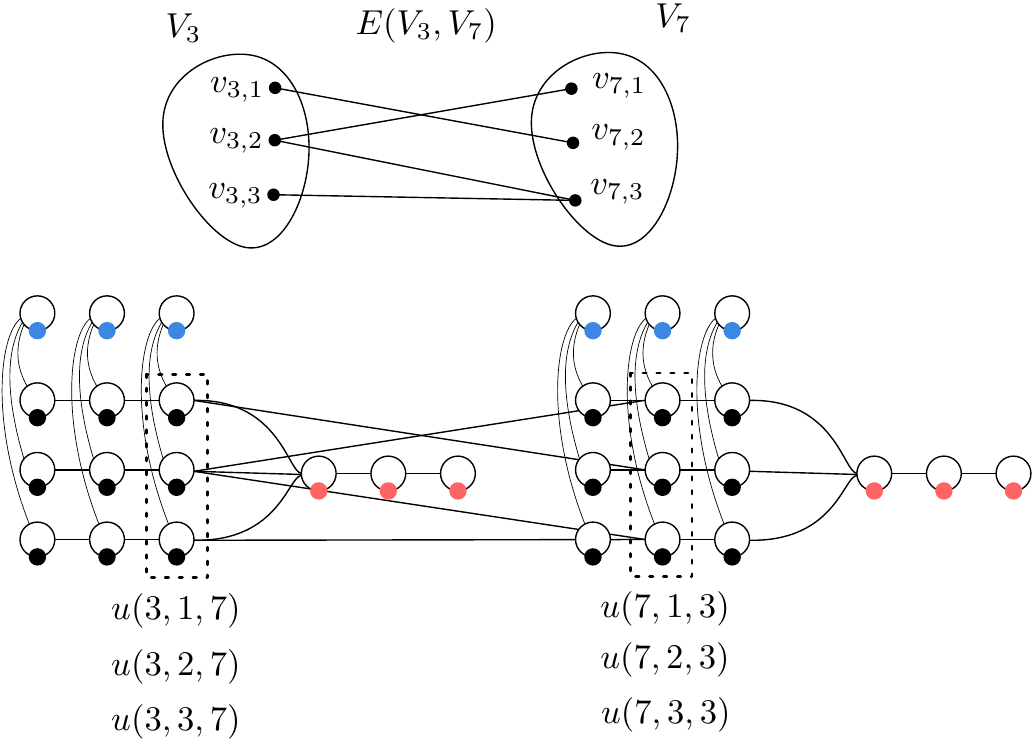}
\caption{The way linkers (in that case, $L_3$ and $L_7$) are assembled together, with $t=3$.}
\label{fig:Linking}
\end{figure}

We complete the construction by adding an edge $u(i,h,j)u(j,h',i)$ whenever $v_{i,h}v_{j,h'}$ is an edge in $E(V_i,V_j)$ (see Figure~\ref{fig:Linking}).
Let $G$ be the graph that we built, and let $I$ be the whole instance of \tok (with the initial position of the tokens).
We claim that $(H,P)$ is a \yes-instance of \msi if and only if $I$ has a solution of length at most $\ell := 16.5k=O(k)$.
Recall that $k$ is even, so $16.5k$ is an integer.

\paragraph{\textbf{Correctness.}} 
$(\Rightarrow)$ First assume that there is a solution $\{v_{1,h_1},  v_{2,h_2}, \ldots, v_{k,h_k}\}$ to the \msi instance.
We perform the following sequence of swaps.
The orderings that \emph{we do not specify} among those swaps are not important, which means that they can be done in an arbitrary fashion.  
In each gadget $L_i$, we first bring $\loc(i,j_3)$ to $b(i,j_3)$, then $\loc(i,j_2)$ to $b(i,j_2)$, and finally $\loc(i,j_1)$ to $b(i,j_1)$, each time passing through the same private path $\U(i,h_i)$.
This corresponds to a total of $12$ swaps per gadget and $12k$ swaps in total.
Note that $\glo(i,j_p)$ is moved to $u(i,h_i,j_p)$.
Now, for each edge $v_{i,h_i}v_{j,h_j}$ of the host graph $H$ (i.e., $u_iu_j \in E(P)$), we swap the tokens $\glo(i,j)$ and $\glo(j,i)$.
By construction of $G$, $u(i,h_i,j)u(j,h_j,i)$ is indeed an edge in $E(G)$, so this swap is legal.
This adds $3k/2$ swaps.
At this point, the token $\glo(j,i)$ is on vertex $u(i,h_i,j)$.
Therefore, we move each token $\glo(j,i)$ to the vertex $a(i,j)$ in one swap.
This corresponds to $3k$ additional swaps.
Observe that it has also the effect of putting the private tokens back to their original private path.
Thus, every token is now well placed.
The overall number of swaps in this solution is $12k+3k/2+3k=16.5k=\ell$.

$(\Leftarrow)$ We now assume that there is a solution $\s$ to \tok of length at most $\ell$.
We define $Y := \{(i,j) \,|\, u_iu_j\in E(P)\}$.
Note that $(i,j)\in Y$ implies $(j,i)\in Y$, and $|Y|=3k$. 
We compute the sum $L(I)$ of the distances \emph{token to destination}.
For any $(i,j) \in Y$, $\loc(i,j)$ is at distance $4$ of its destination $b(i,j)$ (via any private path).
For any $(i,j) \in Y$, $\glo(i,j)$ is at distance $5$ of its destination $a(j,i)$ (following any private path of $L_i$, then an edge to gadget $L_j$, and a last edge to $a(j,i)$).
The rest of the tokens are initially well-placed.
Therefore, $L := L(I)=(4 + 5) \cdot 3k=27k$.
By Proposition~\ref{prop:lower}, the length of any solution for $I$ is at least $13.5k$.

\begin{claim}\label{lem:3k-unhappy-swaps}
In any solution $\s$ for $I$, at least $3k$ initially well-placed tokens have to move.
\end{claim}

\begin{inproof}{\ref{lem:3k-unhappy-swaps}}
There are $3k$ local tokens and each has a disjoint neighborhood from all the others.
Furthermore, all tokens in their neighborhood are private tokens, 
which are already well placed.
\end{inproof}

% Note that a swap with a well-placed token does not decrease the total sum of the distances. Thus, a solution for $I$ is at least of length $16.5k$.
% And except the swaps described in Lemma~\ref{lem:3k-unhappy-swaps} all other swaps must be happy swaps and exactly $3$ private tokens per linker gadget are swapped.

In solution $\s$, let $x$ be the number of swaps between a well-placed token and a misplaced token (in the best case, \xyswaps{-1}{+1}), and $y$ the number of swaps between two well-placed tokens (\xyswaps{+1}{+1}). 
Claim~\ref{lem:3k-unhappy-swaps} implies that $x+2y \geqslant 3k$. 
Those $x+y$ swaps increase the sum of distances \emph{token to destination} by $2y$; its value reaches $L+2y$.
As $\ell \leqslant 16.5k$, there can only be at most $16.5k-(x+y) \leqslant 13.5k+y=\frac{L+2y}{2}$ other swaps.
Therefore, all those swaps shall be happy.
It also implies that in each $\U(i)$ exactly $3$ well-placed tokens move in solution $\s$. A last consequence is that all the swaps strictly worse than \xyswaps{-1}{+1} (that is, \xyswaps{0}{+1} and \xyswaps{+1}{+1}) have to be swaps between two well-placed tokens.

\begin{claim}\label{lem:locality}
In any solution $\s$, no token $\loc(i,j)$ leaves the gadget $L_i$.
\end{claim}
\begin{inproof}{\ref{lem:locality}}
%   Let $N_{i,j}$ be the private tokens in the neighborhood of $a(i,j)$.
%   Recall that these neighborhoods are disjoint and for each $N_{i,j}$ at least one of its tokens must be swapped as otherwise $\loc(i,j)$ cannot reach its destination.
%   Thus every $\loc(i,j)$ makes at most one unhappy swap and this swap must be with a private token.
%   
%   If a local token leaves the gadget, then it had performed an additional unhappy swap --- a contradiction.
% 
It should first be noted that the token $\loc(i,j)$ can only increase its distance to its destination by leaving $L_i$. 
Let $j_1 < j_2 < j_3$ be such that $(i,j_l) \in Y$ for every $l \in [3]$.
% and $l$ be such that $j=j_l$.
The distance of $\loc(i,j)$ to its destination is \emph{its distance to $b(i,j_1)$ plus $l-1$}.
Besides, $\loc(i,j)$ can only leave $L_i$ via a vertex $u(i,h,j')$ with $h \in [t]$ and $(i,j') \in Y$.
From this vertex, it can go to $u(j',h',i)$ for some $h' \in [t]$.
Now, the distance of $\loc(i,j)$ to $b(i,j_l)$ is $2$ if $l=3$, and at least $3$ otherwise.
In both cases, the swap that puts $\loc(i,j)$ cannot be happy.
Therefore, by the consequences of Claim~\ref{lem:3k-unhappy-swaps}, it has to be a swap with a well-placed token.
That means that this swap is at best a \xyswap{0}{+1}.  
This is only possible if it is a \xyswap{+1}{+1} between two well-placed tokens; hence, a contradiction.
\end{inproof}

\begin{claim}\label{lem:consistency}
For every $i \in [k]$, the $3$ tokens of $\U(i)$ which moved in solution $\s$, are in the same $\U(i,h_i)$, for some $h_i \in [t]$.
\end{claim}

\begin{inproof}{\ref{lem:consistency}}
Let $j_1 < j_2 < j_3$ such that $(i,j_1), (i,j_2),$ and $(i,j_3)$ are all in $Y$.
Consider the token $\loc(i,j_2)$.
It first moves to a vertex $u(i,h_i,j_2)$ (for some $h_i \in [t]$).
By Claim~\ref{lem:locality}, its only way to its destination $b(i,j_2)$ is via $u(i,h_i,j_3)$.
This means that the token initially well-placed on $u(i,h_i,j_3)$ is one of those $3$ tokens of $\U(i)$ which moved.
Now, by considering the token $\loc(i,j_1)$, the same argument shows that the three tokens of $\U(i)$ which are moved by solution $\s$ are $u(i,h_i,j_1)$, $u(i,h_i,j_2)$, and $u(i,h_i,j_3)$.
\end{inproof}

We now claim that $\{v_{1,h_1},v_{2,h_2},\ldots,v_{k,h_k}\}$ is a solution to the \msi instance.
Indeed, for any $(i,j) \in Y$, $\glo(i,j)$ has to go to $a(j,i)$.
By Claim~\ref{lem:consistency}, it has to be via vertices of $\U(i,h_i)$ and $\U(j,h_j)$, and there is an edge between those two sets only if $v_{i,h_i}v_{j,h_j} \in E(H)$. 

The graph $G$ has $3(t+2)k$ vertices and $O(t^2k^2)$ edges.
We recall that $\ell = O(k)$.
Therefore, any algorithm solving \tok in time $f(\ell)(|V(G)|+|E(G)|)^{o(\ell / \log \ell)}$, for some computable function $f$, could be used to solve \msi in time $f'(k)(|V(H)|+|E(H)|)^{o(k / \log k)}$; and would contradict the ETH.
This completes the proof of Theorem~\ref{th:param-k-ts}. 
\end{proof}

\section{Token Swapping on nowhere-dense classes of graphs}\label{sec:nowheredense}
As we have seen in Section~\ref{sec:param}, there is little hope for an FPT algorithm for \tok (parameterized by $k$), unless \fpt = \wone. Now let us show that FPT algorithms exist, if we restrict our input to nowhere-dense graph classes.

To define nowhere-dense graphs, first let us introduce a notion of a {\em shallow minor}.
A shallow minor of a graph $G$ at depth $d$ is a subgraph of a graph obtained from $G$ by contracting subgraphs of $G$, each of radius at most $d$, into single vertices, and removing loops and multiple edges. A class $\cal G$ is nowhere-dense if for every $d$ the class of shallow minors at depth $d$ of graphs in $\cal G$ has bounded clique number.
For more information about this topic, we refer the reader to  the comprehensive book of Ne\v{s}et\v{r}il and Ossona de Mendez \citep[Chapter 13]{Sparsity}.

As graphs with bounded degree are nowhere-dense, this result generalizes Proposition~\ref{prop:fpt-delta}.

\FPTnowheredense*

\begin{proof}
If we are able to express  \stok  as a first-order formula, then the result follows immediately from the meta-theorem by Grohe, Kreutzer, and Siebertz \cite{Grohe2014}.

\begin{theorem}
[Grohe, Kreutzer, and Siebertz~\cite{Grohe2014}]
For every nowhere-dense class $C$ and every  $\varepsilon > 0$, every property of graphs definable by a first-order formula $\varphi$ can be decided in time $O(f(|\varphi|,\varepsilon) \cdot n^{1+\varepsilon} )$ on $C$, where $f$ is some function depending only on $\varphi$ and $\varepsilon$.
\end{theorem}

We will define the instance of \stok as a first-order formula $\Phi_{\leq k}$ of size $O(k^4)$. Recall that if the length of an optimal solution is $k$, then at most $2k$ tokens are swapped. In our formula variables will denote vertices of $G$. The relation $edge(x,y)$ denotes the existence of an edge $xy$.
The subsets of possible destinations of tokens will be represented by relation $target(x,y)$, which means that the vertex $y$ is a possible destination for the token initially starting on vertex $x$. Moreover, each token will be identified by its initial position.

Let $\Phi_k$ denote the formula encoding the solution of  \stok  with {\bf exactly} $k$ swaps. If we are interested in a solution using at most $k$ swaps, it is given by $\Phi_{\leqslant k} = \bigvee_{i=1}^k \Phi_i$.

We use variables to represent:
\begin{enumerate}
\item the ``traced'' tokens $t_1,t_2,\ldots,t_{2k}$ that are involved in the solution (some of them may stay intact, if the solution uses less than $2k$ tokens),
\item the final positions $dest_1,dest_2,\ldots,dest_{2k}$ of the ``traced'' tokens ($dest_j$ is the final position of token $t_j$),
\item the swaps $s_1^1,s_1^2,\ldots,s_k^1,s_k^2$ (in the $i$-th swap we exchange the tokens on edge $s_i^1s_i^2$),
\item the tokens that are swapped in the $i$-th swap for $i=1,2,\ldots,k$ -- by $st_i^1,st_i^2$ we denote the tokens that were swapped in the $i$-th swap, i.e., $st_i^p$ denotes the token on vertex $s_i^p$ before performing the $i$-th swap,
\item the positions of ``traced'' tokens in each round -- $pos_{j,i}$ is the vertex, where token $t_j$ is after $i$-th swap.
\end{enumerate}
Now we are ready to present the formula $\Phi_k$.
\allowdisplaybreaks
\begin{flalign}
\Phi_k = &\exists  (t_1, t_2,\ldots,t_{2k}) \label{line:start-def}\\
&\exists  (dest_1, dest_2,\ldots,dest_{2k}) \\
&\exists  (st_1^1,st_1^2, st_2^1,st_2^2,\ldots,st_k^1,st_k^2) \\
&\exists  (s_1^1,s_1^2, s_2^1,s_2^2,\ldots,s_k^1,s_k^2) \\
&\exists  (pos_{1,0},pos_{2,0},\ldots,pos_{2k,0}) \\
&\exists  (pos_{1,1},pos_{2,1},\ldots,pos_{2k,1}) \\
&\exists  (pos_{1,2},pos_{2,2},\ldots,pos_{2k,2}) \\
& \vdots \\
&\exists   (pos_{1,k},pos_{2,k},\ldots,pos_{2k,k}) \label{line:end-def}\\
&\forall(x) (\bigwedge_{j=1}^{2k} x \neq t_j) \to target(x,x) \label{line:not-touched}\\
&\land  \bigwedge_{j=1}^{2k}\bigwedge_{j'=1}^{2k} (j \neq j' \to t_j \neq t_{j'} ) \label{line:different} \\
&\land  \bigwedge_{j=1}^{2k} target(t_j,dest_j)\label{line:good-dest}\\
&\land  \bigwedge_{i=1}^k edge(s_i^1,s_i^2) \label{line:good-swap} \\
&\land  \bigwedge_{j=1}^{2k} pos_{j,0} = t_j \label{line:pos0}\\
&\land  \bigwedge_{j=1}^{2k} pos_{j,k} = dest_j \label{line:posk}\\
&\land  \bigwedge_{i=1}^{k} \left(  \bigvee_{j=1}^{2k}  st_i^1 = t_j \land pos_{j,i} = s_i^1\right) \label{line:synchro1} \\
&\land  \bigwedge_{i=1}^{k} \left( \bigvee_{j=1}^{2k}  st_i^2 = t_j \land pos_{j,i} = s_i^2\right) \label{line:synchro2}\\
&\land  \bigwedge_{i=1}^{k}  \bigwedge_{j=1}^{2k} \left( \lnot (st_i^1 = t_j \lor st_i^2 = t_j) \to pos_{j,i+1} = pos_{j,i} \right) \label{line:stays}\\
&\land  
\bigwedge_{i=1}^{k} 
\bigwedge_{j=1}^{2k}  
\bigwedge_{j'=1}^{2k} \label{line:swaps1} \\
& \left( (j \neq j' \land st_i^1 = t_j \land st_i^2 = t_{j'}) \to (pos_{j,i+1} = pos_{j',i} \land pos_{j',i+1} = pos_{j,i}) \right) \label{line:swaps2}
\end{flalign}
In lines \ref{line:start-def}--\ref{line:end-def} we define the variables. Line \ref{line:not-touched} says that the tokens that are not involved in any swaps are already at feasible positions. Line \ref{line:different} ensures that the traced tokens are pairwise different. Lines \ref{line:good-dest} and \ref{line:good-swap} say that the final positions of traced tokens should be feasible, and we can perform swaps only on edges. In lines \ref{line:pos0} and \ref{line:posk} we synchronize the values of variables $pos_{j,0}$ and $pos_{j,k}$ with variables $t_j$ and $dest_j$. In lines \ref{line:synchro1} and \ref{line:synchro2} we synchronize the values of variables $sp_i^1,sp_i^2$ and $s_i^1,s_i^2$. In line \ref{line:stays} we make sure that the tokens that are not involved in the current swap, stay on their positions. Finally, in line~\ref{line:swaps1} and~\ref{line:swaps2} , we say that the tokens involved in the current swap exchange their positions.
\end{proof}

We derive the following corollary.

\FPTcor*
\bigskip
To see Corollary~\ref{cor:param-k}~{\em(\ref{k+tree})}, recall that bounded-treewidth graphs are nowhere-dense. Therefore by Theorem \ref{th:param-k-lbtw-ts} there exists an algorithm with running time $O(f(k) n^{1+\varepsilon})$, for any $\varepsilon >0$ and treewidth bounded by some constant $c$. Observe that the constant hidden in the big-O notation depends on the constant $c$. In particular $c$ has no influence on the exponent of $n$.

\section{Token Swapping on almost trees}\label{sec:almosttrees}
This section is devoted to the proof of the following theorem.

\HardAlmostTrees*
\begin{proof}
In \ecover, we are given a finite family, denoted by 
$\mathcal S=\{S_1,S_2,\ldots,S_m\}$, of $3$-element subsets of the universe $X=\{x_1, x_2, \ldots, x_n\}$, where $3$ divides $n$.
The goal is to find $n/3$ subsets in $\mathcal S$ that partition (or here, equivalently, cover) $X$. 
The problem can be seen as a straightforward generalization of the \dm problem. This problem is NP-complete and has no $2^{o(n)}$ algorithm, unless the ETH fails, even if each element belongs to exactly 3 triples \cite{Gonzalez85,Bodlaender2015}.
Therefore we can reduce from the restriction of the \ecover problem, where each element belongs to  $3$ sets of~$\mathcal S$, and obviously $|{\cal S}| = |X| = n$.

\paragraph{\textbf{Construction.}} For each set $S_j \in \mathcal S$, we add a \emph{set gadget} consisting of a tree on $10$ vertices (see Figure~\ref{fig:tok-tw2-set-gadget}).
In the set gadget, the four gray tokens should cyclically swap as indicated by the dotted arrows: $g^j_i$ shall go where $g^j_{i+1}$ is, for each $i \in [4]$ (addition is computed modulo 4).
The three black tokens, as usual, are initially well placed.
The three remaining vertices are called \emph{element} vertices. 
They represent the three elements of the set.
The tokens initially on the \emph{element} vertices are called \emph{element} tokens.
For each element of $X$, there are $3$ \emph{element} tokens and $3$ \emph{element} vertices.

\begin{figure}[h!]
\centering
\includegraphics{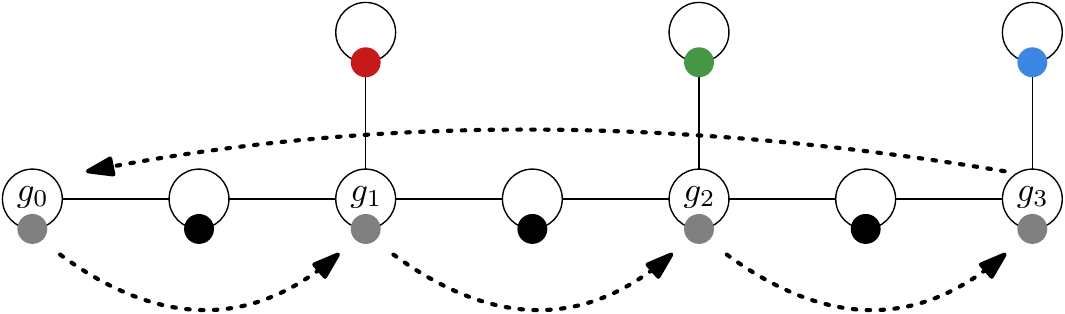}
\caption{The set gadget for 
% $\{$\tikz\fill[red] (0,0) circle (0.1);,\tikz\fill[green] (0,0) circle (0.1);,\tikz\fill[blue] (0,0) circle (0.1);$\}$. 
red, green and blue.
We voluntarily omit the superscript $j$.}
\label{fig:tok-tw2-set-gadget}
\end{figure}

We add a vertex $c$ that is linked to all the \emph{element} vertices of the set gadgets and to all the vertices $g^j_0$.
Each token originally on an \emph{element} vertex should cyclically go to \emph{its next occurrence} (see Figure~\ref{fig:tok-tw2-overall}).
The token initially on $c$ is well placed (it could be drawn as a black token).

\begin{figure}[h!]
\centering
\includegraphics{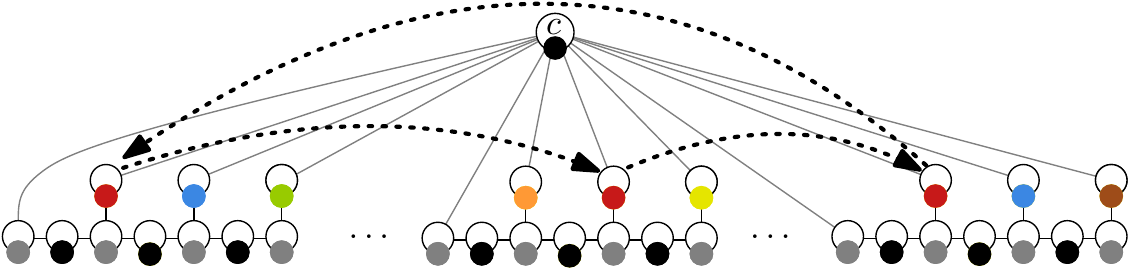}
\caption{The overall picture. Each element appears exactly $3$ times, so there are $3$ red tokens.}
\label{fig:tok-tw2-overall}

\end{figure}

The constructed graph $G$ has $10n+1$ vertices.
If one removes the vertex $c$ the remaining graph is a forest, which means that the graph has a feedback vertex set of size $1$ and, in particular, treewidth $2$.
$G$ has its diameter bounded by $6$, since all the vertices are at distance at most $3$ of the vertex $c$. 
We now show that the instance $\cal S$ of \ecover admits a solution if and only if there exists a solution for our instance of \tok of length at most $\ell := 11 \cdot n/3+9 \cdot 2n/3 + 2n=35n/3=11n+2n/3$.

\paragraph{\textbf{Soundness.}}
The correctness of the construction relies mainly on the fact that there are two competitive ways of placing the gray tokens.
The first way is the most direct. 
It consists of only swapping along the \emph{spine} of the set gadget.
By \emph{spine}, we mean the $7$ vertices initially containing gray or black tokens.
From hereon, we call that \emph{swapping the gray tokens internally}.
\begin{claim}\label{lem:internally}
Swapping the gray tokens internally requires $9$ swaps.
\end{claim}
\begin{inproof}{\ref{lem:internally}}
In $6$ swaps, we can first move $g_3$ to its destination (where $g_0$ is initially).
Then, $g_0$, $g_1$, and $g_2$ need one additional swap each to be correctly placed.
We observe that, after we do so, the black tokens are back to their respective destination.
\end{inproof}

We call the second way \emph{swapping the gray tokens via $c$}.
Basically, it is the way one would have to place the gray tokens if the black tokens (except the one in $c$) were removed from the graph.
It consists of, first (a) swapping $g_0$ with the token on $c$, then moving $g_0$ to its destination, then (b) swapping $g_1$ with the current token on $c$, moving $g_1$ to its destination, (c) swapping $g_2$ with the token on $c$, moving $g_2$ to its destination, finally (d) swapping $g_3$ with the token on $c$ and moving it to its destination.
\begin{claim}\label{lem:via-c}
Swapping the gray tokens via $c$ requires $11$ swaps.
\end{claim}
\begin{inproof}{\ref{lem:via-c}}
Steps (a), (b), and (c) take $3$ swaps each, while step (d) takes $2$ swaps. 
\end{inproof}

Considering that swapping the gray tokens via $c$ takes $2$ more swaps than swapping them internally, and leads to the exact same configuration where both the black tokens and the \emph{element} tokens are back to their initial position, one can question the interest of the second way of swapping the gray tokens.
It turns out that, at the end of steps (a), (b), and (c), an \emph{element} token is on vertex $c$.
We will take advantage of that situation to perform two consecutive happy swaps with its two other occurrences.
By doing so, observe that the first swap of steps (b), (c), and (d) are also happy and place the last occurrence of the \emph{element} tokens at its destination.   

We assume that there is a solution $S_{a_1}, \ldots, S_{a_{n/3}}$ to the \ecover instance.
In the corresponding $n/3$ set gadgets, swap the gray tokens via $c$ and interleave those swaps with doing the two happy swaps over \emph{element} tokens, whenever such a token reaches $c$.
By Claim~\ref{lem:via-c}, this requires $11 \cdot n/3 + 2n$ swaps.
At this point, the tokens that are misplaced are the $4 \cdot 2n/3$ gray tokens in the $2n/3$ remaining set gadgets.
Swap those gray tokens internally.
This adds $9 \cdot 2n/3$ swaps, by Claim~\ref{lem:internally}.
Overall, this solution consists of $29n/3+2n=35n/3=\ell$.   

Let us now suppose that there is a solution $\s$ of length at most $\ell$ to the \tok instance.
At this point, we should observe that there are alternative ways (to Claim~\ref{lem:internally} and Claim \ref{lem:via-c}) of placing the gray tokens at their destination.
For instance, one can move $g_3$ to $g_1$ along the spine, place tokens $g_2$ and $g_3$, then exchange $g_0$ with the token on $c$, move $g_0$ to its destination, swap $g_3$ with the token on $c$, and finally move it to its destination.
This also takes $11$ swaps but moves only one \emph{element} token to $c$ (compared to moving all three of them in the strategy of Claim~\ref{lem:via-c}).  
One can check that all those alternative ways take $11$ swaps or more.
Let $r \in [0,n]$ be such that $\s$ does \emph{not} swap the gray tokens internally in $r$ set gadgets (and swap them internally in the remaining $n-r$ set gadgets).
The length of $\s$ is at least $11r+9(n-r)+2(n-q)+4q=11n+2(r+q)$, where $q$ is the number of elements of $X$ for which \emph{none} occurrence of its three \emph{element} tokens has been moved to $c$ in the process of swapping the gray tokens.
Indeed, for each of those $q$ elements, $4$ additional swaps will be eventually needed.
For each of the remaining $n-q$ elements, only $2$ additional happy swaps will place the three corresponding \emph{element} tokens at their destination.     
It holds that $3r \geqslant n-q$, since the \emph{element} tokens within the $r$ set gadgets where $\s$ does not swap internally represent at most $3r$ distinct elements of $X$.
Hence, $3r+q \geqslant n$.
Also, $\s$ is of length at most $\ell=11n+2n/3$, which implies that $r+q \leqslant n/3$.
Thus, $n \leqslant 3r+q \leqslant 3r+3q \leqslant n$.
Therefore, $q=0$ and $r=n/3$.
Let $S_{a_1},\ldots,S_{a_{n/3}}$ be the $n/3$ sets for which $\s$ does not swap the gray tokens internally in the corresponding set gadgets.
For each element of $X$, an occurrence of a corresponding \emph{element} token is moved to $c$ when the gray tokens are swapped in one of those gadgets.
So this element belongs to one $S_{a_i}$ and therefore $S_{a_1},\ldots,S_{a_{n/3}}$ is a solution to the instance of \ecover.

The ETH lower bound follows from the fact, that the size of constructed graph is linear in $n$.\end{proof}

\section{Variants of Token Swapping on stars, cliques, and paths}\label{sec:trees}
In this section we investigate the complexities of the variants of \tok on very simple classes of graphs.

Let us start with  defining an auxiliary digraph, which will be useful in coping with \ctok.
For an instance of \ctok on a graph $G$, we define the  {\em color digraph} $G^*$, whose vertices are colors of tokens on $G$, and arcs correspond to vertices of $G$.
The vertex $v$ corresponds to the arc $e(v) = cc'$, such that $c$ is the color of $v$ and $c'$ is the color of the token placed in $v$. Note that both loops and multiple arcs are possible. There is a very close relation between color digraphs and Eulerian digraphs.

\begin{observation}\label{obs:eulerian}
The following hold:
\begin{enumerate}[(i)]
\item if $G^*$ is the color digraph of some instance of \ctok, then every connected component of $G^*$ is Eulerian;
\item for every Eulerian digraph $H$ with $n$ edges, and for any graph $G$ with $n$ vertices, there exists an instance of \ctok on $G$, such that its color digraph $G^*$ is isomorphic to $H$.
\end{enumerate}
\end{observation}

\begin{proof}
To see (i), consider a vertex $c$ of $G^*$. Its out-degree is the number of tokens placed on vertices with color $c$. The in-degree of $c$ is the number of tokens in color $c$. Thus the in-degree is equal the out-degree, from which (i) follows.

Now, to see (ii), consider a vertex $c$ of $G^*$, let $d$ be its out-degree (equal to the in-degree, as $G^*$ is Eulerian). Then in $G$ give the color $c$ to any $d$ vertices. Moreover, for each arc $cc'$ in $G^*$ we place a token in color $c'$ on a vertex in color $c$. We repeat this for every vertex $c$ in $G^*$, obtaining an instance of \ctok, whose color digraph is exactly $G^*$.
\end{proof}

Now consider a solution $\s$ for the instance of  \ctok  in $G$ and fix the destinations of tokens according to $\s$. We observe that the cycles in the permutation defined by these destinations correspond to circuits in $G^*$. Thus, when trying to find a solution for an instance of \ctok, we will first try to fix appropriate destinations (by analyzing circuits in $G^*$), and then we will solve the instance of \tok.

\subsection{Stars}

To prove the next theorem we will use the following result by Pak \cite{PAK1999}. We state in the language of tokens and swaps, although the original motivation of Pak was sorting a permutation by transpositions with the first element.

\begin{lemma}[Pak \cite{PAK1999}] \label{lem-pak}
Let $I$ be an instance of \tok on a star with $n$ leaves, with the initial configuration of tokens $\pi$.
If the decomposition of $\pi$ into cycles consists of one cycle involving the central vertex, $m$ cycles of length at least $2$, and $b$ cycles of length $1$, then the length of an optimal solution to $I$ is $n+m-b$.
\end{lemma}

\begin{theorem}\label{th:stars-cts}
\ctok can be solved in polynomial time on stars.
\end{theorem}

\begin{proof}
Let $G$ be a star with center $v_0$ and leaves $v_1,v_2,\ldots,v_n$. The color of the vertex $v$ will be denoted by $c(v)$. Also, let $c_0 := c(v_0)$.

First, suppose that there exists a leaf $v$, such that the token $t$ that is initially placed there has color $c(v)$ as well. 
Let $\s$ be an optimal solution and consider a permutation $\pi$ of tokens given by $\s$. 
We want to show that $\pi(v) = v$. Using the solution of Pak~\cite{PAK1999}, this implies that $t$ is never swapped.

For the purpose of contradiction, suppose $\pi(v) \neq v$. Then, there exists a token $t'$ initially on vertex $u$ with $\pi(u) = v$  and a vertex $w$ with $\pi(v) = w$. In other words, token $t$ ends at vertex $w$. So neither $u,v$ nor $w$ is involved in a $1$-cycle, but all three vertices must have the same color.
Thus we can alter this solution to a new permutation $\pi'$, by setting $\pi(v) = v$ and $\pi(u) = w$.
This increases the number $b$ of $1$-cycles by $1$ and the number $m$ stays the same.
This contradicts the optimality of $\s$ by Lemma~\ref{lem-pak} and we conclude 
$\pi(v)=v$.
% and it is easy to observe that the edge $v_0\,v$ does not appear in $\s$.
% 
% (thus $v$ is not a $1$-cycle in $\pi$), then we could obtain a shorter solution as follows:
% by exchanging the destinations $t$ and the token whose destination is $v_i$. 
Thus for any leaf $v$ with a token $t$ of color $c(v)$ holds, that the solution does not change
after removal of $v$.
% We can obtain an equivalent instance of \ctok by removing $v$. 
Thus from now on we assume that no leaf $v$ contains a token colored with color $c(v)$.

Consider the color digraph $G^*$.  By the previous paragraph, we observe that with just one possible exception $c_0c_0$, it has no loops. Let $C_0,C_2,\ldots,C_m$ be the connected components of $G^*$, and let $c_0 \in C_0$. Moreover, for $i\geq 0$, by $p_i$ we denote the number of arcs in $C_i$. By Observation \ref{obs:eulerian}(i), the edges of $G^*$ can be decomposed (in polynomial time) into $m+1$ circuits (Eulerian circuits of its connected components). 

Let $e(v^i_1),e(v^i_2),\ldots,e(v^i_{p_i})$ be such a circuit for $C_i$, also we assume that $v^0_1 = v_0$ (i.e. we start the circuit for $C_0$ with the arc corresponding to $v_0$).
We construct the swapping strategy $\s$ by concatenating sequences $\s_i$, defined as follows:
\[
\s_i = 
\begin{cases}
v_0v^0_2,v_0v^0_3,\ldots,v_0v^0_{p_0} & \text{ for } i=0,\\
v_0v^i_1,v_0v^i_2, v_0v^i_3,\ldots,v_0v^i_{p_i} & \text{ for } i>0.
\end{cases}
\]
It is straightforward to verify that $\s$ is a solution for our problem and its length is $n+m$. We claim this solution is optimal.

To see this, consider any solution $\s'$. Let  us consider the instance of \tok obtained by fixing the destinations of all tokens, according to $\s'$. Let $q_0,q_1,\ldots,q_{m'}$ be the cycles in the permutation given by the destinations, and assume $q_0$ contains vertex $v_0$. By Lemma \ref{lem-pak}, the length of the optimal solution of this instance of \tok is exactly $n+m'$. We observe that the set of colors of vertices in each cycle has to be entirely contained in one of the components $C_i$, so $m' \geq m$, thus the length of $\s'$ is at least $n + m' \geq n+m$, which completes the proof.
\end{proof}

 \begin{theorem}\label{th:stars-sts}
 On stars, \stok remains \emph{NP}-hard and cannot be solved in time $2^{o(n)}$ unless the ETH fails, even for target sets of size at most $2$.
 \end{theorem}

\begin{proof}
We will reduce from the \textsc{Directed Hamiltonian Cycle} problem restricted to digraphs with out-degree at most 2, which is known to be NP-complete \cite{Plesnik79}.
Moreover, it follows from the proof that the problem cannot be solved in time $2^{o(n)}$, unless the ETH fails (the original proof
considers planar instances, but if we drop the planarity assumption, we obtain claimed lower bound).

Let $G=(V,E)$ be a digraph with all out-degrees at most 2, we can assume it has no loops. We will construct an instance $(G'=(V',E'), D)$ of \stok with $|D(v)| \leqslant 2$ for all $v \in V'$, that has a solution of length at most $n+1$ if an only if $G$ has a Hamiltonian cycle.

The set $V'$ is equal to $V \uplus \{c\}$ where $c$ is the center of the star, and the leaves are the vertices of $G$. For each $v \in V' \setminus \{c\}$, we set $D(v)=N_G(v)$ (the set of out-neighbors of $v$ in $G$) and $D(c)= \{c\}$.

Suppose $G$ has a Hamiltonian cycle $v_1,v_2,v_3,\ldots,v_n$ (with $v_1$ adjacent to $v_n$). It is easy to observe that the sequence $cv_1,cv_2,\ldots,cv_n,cv_1$ of edges is a solution for  \ctok  with length $n+1$.

On the other hand, suppose there is a solution $\s'$ for  \stok  of length at most $n+1$. Since $G$ has no loops, every token starting at $v \in V$ must be moved to $c$ at some point. Moreover, in the last swap we have to bring the token starting at $c$ back to this vertex. Thus every feasible solution uses at least $n+1$ swaps, which implies that the length of  $\s'$ is exactly $n+1$; let $\s' = cv_1,cv_2,\ldots,cv_n,cv_{n+1}$. Moreover, we have $v_1 = v_{n+1}$ and $v_i \neq v_j$ for all $1 \leq i < j \leq n$. Thus we observe that $v_1,v_2,v_3,\ldots,v_n$ is a Hamiltonian cycle in $G$.
\end{proof}

\subsection{Cliques}

If $G$ is a complete graph, then the optimal solution for \tok is $n$ minus the number of cycles in the permutation given by initial positions of tokens \cite{Cayley}. Thus, the problem is solvable in polynomial time. On the other hand, we can show that \ctok  is NP-complete on cliques.
Before we prove it, let us prove an auxiliary lemma.
In the {\sc Directed Triangle Decomposition} we are given a digraph $H=(V,A)$, and we ask whether the arc set $A$ can be decomposed into disjoint directed triangles.

\begin{lemma}\label{lem:decomposition}
{\sc Directed Triangle Decomposition} is \emph{NP}-complete, even if the input digraph $H=(V,A)$ is Eulerian and has no 2-cycles. Moreover, it cannot be solved in $2^{o(|A|)}$, unless the ETH fails.
\end{lemma}

\begin{proof}
For a given \sat formula $\Phi$ with $N$ variables and $M$ clauses, we will construct a digraph $H$, which can be decomposed into triangles if and only if $\Phi$ is satisfiable.

The main part of the construction is essentially the same as the construction of Holyer~\cite{Holyer81}, used to show NP-hardness of decomposing the edge set of an undirected graph into triangles (or, more generally, $k$-cliques). Thus we will just point out the modifications and refer the reader to the paper of Holyer for a complete description.

We observe that by the proper adjustment of constants the graph $G_3$ constructed by Holyer can be made three-partite (see also Colbourn \cite{COLBOURN198425}). Let $A,B,C$ denote the partition classes. We obtain $H$ by orienting all edges of $G_3$, according to the following pattern $A \to B \to C \to A$. Note that clearly $H$ has no 2-cycles.

Consider a vertex $v$ of $G_3$. Without loss of generality assume $v \in A$. We note that exactly half of the neighbors of $v$ are in $B$, and the other half are in $C$. This implies that $H$ is Eulerian.

We also point out that the number of arcs in $H$ is linear in the number of vertices.
Moreover, if we make the size of each variable gadget proportional to the number of occurrences of this variable in $\Phi$ (instead of proportional to $M$, as in the original proof), we obtain that $|A| = O(N+M)$. This shows that an existence of a subexponential (in $|A|$) algorithm for our problem contradicts the ETH.
\end{proof}

 \begin{theorem}\label{th:cliques-cts}
 On cliques, \ctok remains \emph{NP}-hard and cannot be solved in time $2^{o(n)}$, unless the ETH fails.
 \end{theorem}

\begin{proof}
We reduce from {\sc Directed Triangle Decomposition}.
Let $H$ be an Eulerian directed graph with $n$ arcs, having no 2-cycles.
Consider an instance of \ctok  on $G = K_n$, such that $H$ is its color digraph (it exists by Observation \ref{obs:eulerian}(ii)). We claim that there exists a solution for this instance of length at most $2n/3$ if and only if the arc set of $H$ can be decomposed into directed triangles (see Lemma \ref{lem:decomposition}).

Suppose that the arc set of $H$ can be decomposed into $n/3$ triangles. The vertices of $G$ corresponding to the edges of the $i$-th triangle, are $v^i_1,v^i_2,v^i_3$. We construct the solution $\s$ by concatenating sequences $v^i_1v^i_2,v^i_1v^i_3$ for $i=1,2,\ldots,n/3$. It is easy to verify that $\s$ is a solution and its length is $2n/3$.

So now suppose we have a solution $\s$ of length at most $2n/3$. Recall that the length of any solution $\s'$ is at least $n$ minus the number of cycles in the permutation obtained by fixing the destinations of tokens according to $\s'$. Thus the number of cycles in the permutation given by $\s$ is at least $n/3$. Since these cycles correspond to circuits in the color digraph $H$, and $H$ has no 2-cycles, this is only possible if the arcs of $H$ can be decomposed into triangles.
\end{proof}

It is interesting to point out that if $G$ is a clique, then the presence of many cycles in the permutation of tokens yields a short solution for \tok, while for the case when $G$ is a star, the situation is opposite.

Theorems \ref{th:stars-sts} and \ref{th:cliques-cts} can be used to show a slightly more general hardness result.
A class $\cal G$ of graphs is {\em hereditary}, if for any $G \in \cal G$ and any induced subgraph $G'$ of $G$ we have $G' \in \cal G$.
We say that that a class $\cal G$ of graphs has unbounded degree, if for every $d \in \mathbb{N}$ there exists $G \in {\cal G}$, such that $\Delta(G) \geq d$.

\begin{theorem} \label{thm:unbounded-deg}
Let $\cal G$ be a hereditary class containing an infinite number of connected graphs with unbounded degree.
 \stok  is \emph{NP}-complete, when restricted to graphs from $\cal G$. Moreover, if there exists an algorithm solving  \stok  in time $2^{o(n)}$ for every graph in $ G \in {\cal G}$ with $n$ vertices, then the ETH fails.
\end{theorem}

\begin{proof}
We shall reduce from {\sc Directed Hamiltonian Cycle} in digraphs with out-degree at most 2. Let $H$ be such a digraph with $n$ vertices. 

First, assume that $K_{1,n} \in {\cal G}$. Then we are done by Theorem \ref{th:stars-sts}.
So assume that $K_{1,n} \notin {\cal G}$. Since $\cal G$ is hereditary, we know that $K_{1,n'} \notin {\cal G}$ for any $n' \geq n$.
Since decomposing the arc set of an Eulerian digraph with no 2-cycles into directed triangles is NP-complete (see Lemma \ref{lem:decomposition}), there exists a polynomial reduction from  {\sc Directed Hamiltonian Cycle} to this problem. Consider the digraph $H^*$ obtained with this reduction. Its arc set can be decomposed into triangles if and only if $H$ has a Hamiltonian cycle. Let $m$ denote the number of edges in $H^*$ and set $N = \max(m,n)$.

By Ramsey theorem \cite{Ramsey1987} (see also Erd\H{o}s, Szekeres \cite{Erdos1987}) we know that there exists an absolute constant $c$ such that every graph with more than $c \cdot 4^N$ vertices has either a clique or an independent set of size $N$.

Since $\cal G$ has unbounded degree, there exists a graph $G \in \cal G$, such that $\Delta(G) \geq c \cdot 4^N$. Let $v$ be a vertex of $G$ with degree at least $c \cdot 4^N$ and let $G'$ be a subgraph of $G$ induced by the neighborhood of $v$. 
If $G'$ has an independent set $U$ of size $N$, then $G[U \cup \{v\}] \sim K_{1,N}$, so we obtain a contradiction (recall that $\cal G$ is hereditary). Thus $G'$ has a subset $C$ inducing a clique of size $N$. Since $\cal G$ is hereditary and $N \geq m$, we obtain that $K_m \in \cal G$. Thus we can use the construction from Theorem \ref{th:cliques-cts}.
\end{proof}

\subsection{Paths}

Finally, we turn our attention to paths.
 
 \begin{theorem}\label{th:paths-cts}
 \ctok can be solved in polynomial time on paths.
 \end{theorem}

\begin{proof}
Let $c$ be the color of the vertex $v$ at the left end of the path.
Let $t$ be the leftmost token with color $c$.
It is clear that no optimal solution contains a swap involving two tokens of the same color, so in any optimal solution the token $t$ will end up in $v$.
Repeat this argument with the second leftmost vertex, and so on.
This way we fix the destinations for all tokens, obtaining an equivalent instance of \tok, which can be solved in polynomial time (see \cite{MiltzowNORTU16}).
\end{proof}

Now we will discuss the complexity of \stok on paths. We want to point out an equivalent, interesting formulation of this problem. Consider an instance $\cal I$ of  \stok  defined on a path with $n$ vertices $v_1,v_2,\ldots,v_n$. For a vertex $v_i$, let $t_i$ denote the token initially placed on $v_i$, and let $D(t_i)$ denote the set of possible destinations of $t_i$. Now consider a bipartite graph $G$ with bipartition classes $\{v_1,v_2,\ldots,v_n\}$ and $\{t_1,t_2,\ldots,t_n\}$. The edge $t_iv_j$ is present in $G$ if and only if $v_j \in D(t_i)$. Fix two distinct vertical lines $\ell$ and $\ell'$ on a plane and fix the positions of vertices of $G$ on these lines; $v_1,v_2,\ldots,v_n$ lie on $\ell$ (in this ordering from top to bottom), and $t_1,t_2,\ldots,t_n$ lie on $\ell'$ (also in this ordering from top to bottom); see \autoref{fig:bipartite}.

\begin{figure}[htbp]
 \centering
 \includegraphics{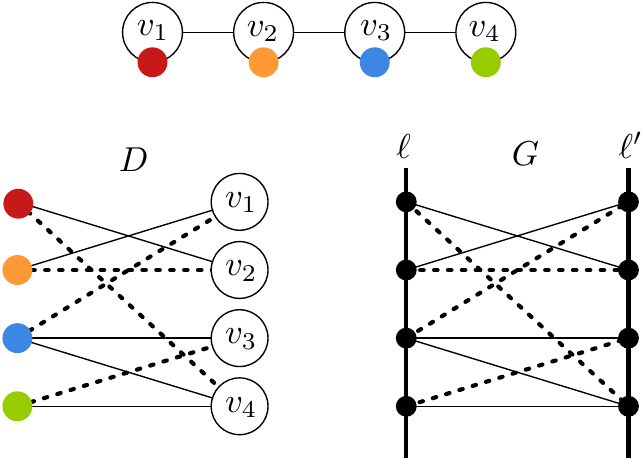}
 \caption{A bipartite graph $G$ constructed from an instance of \stok on a path. 
 %\prz{Till, can you please provide a picture?}
 }
 \label{fig:bipartite}
\end{figure}

Consider a feasible solution $\s$ of $\cal I$ and let $\sigma$ be the permutation assigning destinations to tokens, according to $\s$.
Since after fixing the destinations we obtain an instance of \tok, which is polynomially solvable on paths, we observe that each feasible solution $\s$ for $\cal I$ corresponds to a perfect matching in $G$ (and vice versa).

Recall that the number of swaps required to solve an instance of \tok on a path is equal to the number of inversions in the initial permutation of tokens. Suppose there is such an inversion in $\sigma$, i.e. $\sigma(t_i) > \sigma(t_j)$ for some $i < j$. Observe that this is exactly equivalent to saying that the edges $t_i\sigma(t_i)$ and $t_j\sigma(t_j)$ of $G$ cross (see \autoref{fig:bipartite}).

So let us formally define the problem \bip, which is equivalent to \stok on a path.
The instance of \bip is $(G,k)$, where $k$ is an integer and $G$ is a bipartite graph with $n$ vertices in each bipartition class. Moreover, the vertices of $G$ are positioned on two parallel lines, one for each bipartition class. We can also assume that $G$ has at least one perfect matching.
The problem asks if $G$ has a perfect matching with at most $k$ pairwise crossing pairs of edges.

The problem \bip (and thus also \stok on paths) was recently shown to be NP-hard by Gu\'{s}piel~\cite{Guspiel}.

\begin{theorem}[Gu\'{s}piel~\cite{Guspiel}] \label{thm:paths}
 \stok  remains NP-hard for paths, even if each token has at most 2 possible destinations, and each vertex is a destination of at most 2 tokens. Moreover, the problem cannot be solved in time $2^{o(n)}$ (where $n$ is the number of vertices of the path), unless the ETH fails.
\end{theorem}

This result allows us to generalize \autoref{thm:unbounded-deg} to all hereditary classes of graphs.

\begin{theorem} \label{thm:all-hereditary}
Let $\cal G$ be a hereditary class containing an infinite number of connected graphs.
 \stok  is \emph{NP}-complete, when restricted to graphs from $\cal G$. Moreover, if there exists an algorithm solving  \stok  in time $2^{o(n)}$ for every graph in $ G \in {\cal G}$ with $n$ vertices, then the ETH fails.
\end{theorem}

\begin{proof}
If $\cal G$ has unbounded degree, then the claim holds by \autoref{thm:unbounded-deg}. On the other hand, if there is a constant $d$, such that $\Delta(G) \leq d$ for all $G \in {\cal G}$, then  $\cal G$ contains all paths.
Indeed, let $n$ be an integer and let $G \in {\cal G}$ be a graph with at least $n \cdot d^n$ vertices (it always exists, since $\cal G$ is infinite). Run a BFS algorithm on $G$, starting from an arbitrary vertex, and consider the obtained BFS-layers. The number of such layers is at least $n$, so $G$ contains $P_n$ as an induced subgraph. Since $\cal G$ is hereditary, we have $P_n \in {\cal G}$. The claim follows by \autoref{thm:paths}.
\end{proof}

\section{Conclusion}\label{sec:conclusion}
We conclude the paper with several ideas for further research.
First, we believe that it would be interesting to fill the missing entries in Table  \ref{conc:summary2}. In particular, we conjecture that \tok remains NP-complete even if the input graph is a tree.

Another interesting problem is the following. 
By Miltzow {\em et al.} \cite[Theorem~$1$]{MiltzowNORTU16} (see also Proposition \ref{prop:exact}), \tok can be solved in time $2^{O(n \log n)}$, and there is no $2^{o(n)}$ algorithm, unless the ETH fails.
We conjecture that the lower bound can be improved to $2^{o(n \log n)}$.
It would also be interesting to find single-exponential algorithms for some restricted graph classes, such as graphs with bounded treewidth or planar graphs.

Finally, to prove Corollary \ref{cor:param-k}, we use the powerful and very general meta-theorem by Grohe, Kreutzer, and Siebertz~\cite{Grohe2014}. It would be interesting to obtain elementary FPT algorithms for planar graphs and graph with bounded treewidth (or even trees), just as we did for graphs with bounded degree.

\end{document}